\documentclass[reqno]{amsart}
\usepackage{geometry}                
\usepackage[parfill]{parskip}    
\usepackage{graphicx}
\usepackage{amssymb}
\usepackage{epstopdf}
\usepackage{esint}
\usepackage{mathrsfs} 
\usepackage[english]{babel}
\usepackage{amssymb,
enumerate}
\usepackage{color}
\usepackage[latin1]{inputenc}
\usepackage[T1]{fontenc}
\usepackage{amsfonts}
\usepackage{amssymb}
\usepackage{amsmath}
\usepackage{amsthm}
\usepackage{graphicx,xcolor}
\usepackage{afterpage}
\usepackage[colorlinks=true, linkcolor=red, citecolor=cyan, urlcolor=green]{hyperref} 
\usepackage{bbm}
\usepackage{latexsym}

\def\be{\begin{equation}}
\def\ee{\end{equation}}
\def\bea{\begin{eqnarray}}
\def\eea{\end{eqnarray}}
\def\bel{\begin{align}}
\def\el{\end{align}}
\def\ni{\noindent}
\def\nn{\nonumber}

\def\R{\mathbb{R}}
\def\s{\sigma}
\def\a{\alpha}

\def\b{\beta}

\def\t{\tau}

\def\N{\mathbb{N}}

\newcommand{\ie}{\textit{i.e. }}

\newcommand{\meanv}[1]{\left\langle#1\right\rangle}

\newcommand{\nocontentsline}[3]{}
\newcommand{\tocless}[2]{\bgroup\let\addcontentsline=\nocontentsline#1{#2}\egroup}

\DeclareMathSymbol{\leqslant}{\mathalpha}{AMSa}{"36} 
\DeclareMathSymbol{\geqslant}{\mathalpha}{AMSa}{"3E} 
\DeclareMathSymbol{\eset}{\mathalpha}{AMSb}{"3F}     
\renewcommand{\leq}{\;\leqslant\;}                   
\renewcommand{\geq}{\;\geqslant\;}                   

\newcommand{\dd}{\,\text{\rm d}}

\DeclareMathOperator{\RS}{RS}

\newcommand{\E}[1]{E\left[#1\right]}

\def\ie{\textit{i.e. }}

\def\a{\alpha}

\def\b{\beta}

\def\s{\sigma}
\def\t{\tau}

\def\R{\mathbb{R}}

\def\E{E}

\theoremstyle{plain}
\newtheorem{theorem}{Theorem}[section]
\newtheorem{lemma}[theorem]{Lemma}

\newtheorem{definition}[theorem]{Conjecture}

\numberwithin{equation}{section}

\title[RS free energy of deep RBM\MakeLowercase{s}]
{Minimax formula for the replica symmetric free energy of deep restricted Boltzmann machines}

\author{Giuseppe Genovese}
\address{Institute of Mathematics, University of Zurich, Winterthurerstrasse 190, 8057 Zurich, Switzerland.}
\email{giuseppe.genovese@math.uzh.ch}

\date{\today}                     

\begin{document}
\maketitle

\begin{abstract}
We study the free energy of a most used deep architecture for restricted Boltzmann machines, where the layers are disposed in series. Assuming independent Gaussian distributed random weights, we show that the error term in the so-called replica symmetric sum rule can be optimised as a saddle point. This leads us to conjecture that in the replica symmetric approximation the free energy is given by a $\min\max$ formula, which parallels the one achieved for two-layer case. 
\vspace{0.5cm}

\ni {\bf MSC:} 82D30, 49K35.
\end{abstract}

\section{Introduction}\label{sect:intro}

The deep restricted Boltzmann machine (dRBM) is a widely studied generative model, first introduced in \cite{deep}, in which many RBMs are piled up in a serial architecture. Indeed in many practical applications RBMs fail to model complex data distributions without a careful choice of the weight initialisation in the learning algorithm and augmenting the depth aims chiefly to increase the representational power of the model. Similarly to deep neural networks, the idea of dRBMs is that the multi-layer architecture increases the level of nonlinearity of the model and allows high-order representations in the hidden layers apt to capture higher-order correlations of the data. 

Despite the underlying bipartite structure permitting Gibbs sampling to work efficiently, dRBMs are hard to train, which makes other generative models preferred in practice. The typical gradient-based algorithms as for instance contrastive divergence \cite{CD,pCD} get more easily stuck in poorly representative local maxima of the log-likelihood, due to the more complex landscape. Many improvements of the standard algorithms have been proposed \cite{deep2,deep3}, but the problem is substantially difficult due highly non-convex structure of the log-likelihood (a feature of RBMs, increasing with the depth). One major source of difficulty in the log-likelihood is the logarithm of the partition function of the model (physicist's free energy, in this paper we adopt this terminology). Indeed its lack of convexity makes gradient ascent based optimisation algorithms and Montecarlo approximation methods very sensitive to the initialisation. 

The aim of this note is to obtain some more information of the convexity (or lack thereof) of the free energy. We study the replica symmetric (RS) sum rule of the free energy of dRBMs, assuming the weight matrix to be random, with independent standard Gaussian entries, as customary in the spin glass literature. Such a sum rule is obtained by Gaussian interpolation between the energy term of the dRBM and some independent random biases. The biases are taken to be centred Gaussian random variables, whose variances can be seen as Lagrange multipliers which one can optimise on. This optimisation is in fact the main focus of the paper: we show that it yields a $\min\max$ variational formula similar the one proved for the standard RBM in \cite{NN,bip}. 

Next we introduce the objects we will deal with and state precisely the main result. 

\subsection{The model}
The model is defined as follows. Let $N\in\N$ denote the total number of units and $\nu\geq 2$ the number of layers, indexed by $x=1,\dots,\nu$. 
The $x$-th layer has $N_x$ units, with $\sum_x N_x =N$ and $\a_x:=\lim_N N_x/N\in(0,1)$, with $\sum_x\a_x=1$. We will often use $\a_x$ to denote the ratio $N_x/N$ also at finite size with a small abuse of notation. A total configuration is indicated by $\s$ and the $x$-th layer configuration by $\s^{(x)}$. 

We assume the units $\s^{(x)}$ for all $x\in[\nu]$ to be independent random variables with Bernoulli $\pm1$ a priori distributions, which we bias independently by the entry-wise constant vectors $b^{(x)}:=(b^{(x)},\ldots,b^{(x)})\in\R^{N_x}$, $b^{(x)}\in \R$ (with a little abuse of notations we indicate the vector and its constant components with the same symbol). Expectation values w.r.t. the $x$-th prior distribution is denoted by $\widehat E_{\s^{(x)}}$. Mostly we will omit the biases in the notation. 

\ni Two consecutive layers $x,x+1$ ($x\in[\nu-1]$) interact via the RBM Hamiltonian
\be\label{eq:H}
H^{(x)}_N(\s^{(x)},\s^{(x+1)}):=-\sum_{\substack{i\in[N_x]\\j\in[N_{x+1}]}}\frac{\xi_{ij}^{(x)}}{\sqrt{N}}\s_i^{(x)}\s_j^{(x+1)}\,,
\ee
where the $\{\xi_{ij}^{(x)}\}_{\substack{x\in[\nu-1]\\i\in[N_x]\\j\in[N_{x+1}]}}$ are $\mathcal{N}(0,1)$ i.i.d. quenched r.vs. 

The multilayer model is defined by a combination of RBM Hamiltonians:
$$
H_N(\s)=\sum_{x\in[\nu-1]} H^{(x)}_N(\s^{(x)},\s^{(x+1)})\,.
$$

The Gibbs posterior distribution for any $\b\geq0$ reads
\be
p_{\b}(\s;\xi,b):=\frac{\exp\left(-\b H_N(\s)+\sum_{x\in[\nu]} (b^{(x)},\s^{(x)})\right)}{\widehat E_{\s^{(1)}\ldots\s^{(\nu)}}\left[\exp\left(-\b H_N(\s)+\sum_{x\in[\nu]} (b^{(x)},\s^{(x)})\right)\right]}\nn\\
\ee
where 
we denote by $(\cdot,\cdot)$ the inner product of $\R^{d}$ (regardless of the dimension $d$). The free energy $A_N$ is defined as follows
\be
A_N(\b):=\frac1N\log \widehat E_{\s^{(1)}\ldots\s^{(\nu)}}\left[\exp\left(-\b H_N(\s)+\sum_{x\in[\nu]} (b^{(x)},\s^{(x)})\right)\right]\,.
\ee 

The main interest is in computing the free energy in the thermodynamic limit, that is
$$
N_1,\ldots,N_\nu\to\infty\quad\mbox{such that }\quad N_x/N\to\a_x \quad \forall x\in[\nu]\,. 
$$
We indicate this limiting procedure simply as $\lim_N$. 
The existence of the limit of the free energy is at the moment an open mathematical problem. Assuming it exists however it must be self-averaging, by the standard argument using concentration of Lipschitz functions \cite{tsirelson}. Therefore we can equivalently study the limit of the averages of $A_N$.

Central objects are the overlaps, \ie normalised inner products, of each couple of configurations of the $x$-th layer $\s^{(x)},\t^{(x)}$
$$
R^{(x)}=R^{(x)}(\s^{(x)},{\t}^{(x)}):=\frac{(\s^{(x)},{\t}^{(x)})}{N_x}\,.
$$

Taking two different points $\s,\t\in\{-1,1\}^N$ let us compute the covariance of the energy in terms of the overlaps (assume for simplicity $b^{(x)}=0$ for all $x\in[\nu]$). We have
$$
E[H_N(\s)H_N(\t)]=N\sum_{|x-y|=1} \a_x\a_yR^{(x)}R^{(y)}\,. 
$$
This quantity is not positive definite (the overlaps can take both positive and negative values), which is a major source of difficulties when comparing dRBMs to the Sherrington-Kirkpatrick model. As it will be clear later on, such a lack of positivity of the covariance of the Hamiltonian yields a non-convex variational principle for the free energy (see (\ref{eq:error}) and (\ref{eq:Err-minmax}) below).

For simplicity we present our argument for $\pm1$ Bernoulli priors. We expect the extension to bounded units to be straightforward and that further extensions to sub-Gaussian units are possible, but technically more involved. Indeed two generic sub-Gaussian units coupled by the Hamiltonian \eqref{eq:H} typically make the model ill defined for large $\b$. Therefore one has to regularise the posterior distribution as discussed in \cite{sferico}, and similarly also the interpolating factors introduced below (as in \cite{NN}). These regularisations introduce a number of technicalities, irrelevant for the main message we aim to give here.
Also, as long as the units are bounded, the choice of the Gaussian distribution for the random weights is not so stringent, as the same argument of \cite{univ} applies and universality of the free energy can be proven for a large class of independent random weights. Most interesting are extensions to dependent random weights \cite{mezard}, yet rather unexplored at the moment.

Also, we deal with small biases, \ie all the $b^{(x)}$ are bounded by a fixed constant (see Theorem \ref{TH:minmax}). Dealing with possibly large biases is in fact easier but it requires a procedure which is similar, yet not exactly the same as the one presented in the sequel. For clarity, we focus here only on the most interesting and challenging case. 

In what follows we will conveniently separate the layers into two disjoint subsets. To fix the ideas let us set $H:=2\N\cap[\nu]$ and $V:=(2\N-1)\cap[\nu]$ respectively the layers in an even or odd position in increasing order. Note that units within the layers in $V$ are conditionally independent w.r.t. $p_{\b}(\s;\xi,b)$ given the layers of $H$ (and viceversa), reflecting the aforementioned bipartite structure of the dRBM. 

Throughout $\eta$ (possibly labeled by one or more indices) denotes a $\mathcal N(b,\s)$ random variable, whose expectation is always denoted by $E_{b,\s}$, that is in the univariate case
$$
E_{b,\s}[f]=\int f(x)\frac{e^{-\frac{(x-b)^2}{2\s}}}{\sqrt{2\pi\s}}dx
$$
 (note that $\lim_{\s\to0}E_{b,\s}[f]=f(b)$ for sufficiently regular $f$). When different $\eta$s are i.i.d. $\mathcal N(b,\s)$ we will use a unique symbol $E_{b,\s}$ to denote their joint expectation.  
Averages on the $\mathcal N(0,1)$ quenched weights $\xi^{(x)}_{ij}$ are simply indicated by $E$.

\subsection{Main result and RS approximation}

The aim of this paper is to characterise the RS approximation of the free energy of deep RBMs in terms of a non-convex variational formula. The main motivation for that is to take a step towards the understanding of the free energy. Albeit its exact form will probably involve replica symmetry breaking, we believe that the structure described here should persist also in this more complicated scenario. For a further elaboration on that, see Conjecture \ref{CJ:main} below and related discussion.

The starting point of our considerations is a sum-rule for the free energy as the one obtained for the Sherrington-Kirkpatrick model by Guerra in \cite{sum-rules} and later developed in different contexts. Let $t\in[0,1]$, $q:=(q_1\,,\ldots\,,q_\nu)\in[0,1]^\nu$. Let further for $x\in{\nu}$ $\{\eta^x_j\}_{j\in[N_x]}$ be i.i.d. $\mathcal{N}(0,1)$. We set
\bea
H'_N(\s)&:=&- \sum_{x=1}^\nu \sqrt{\sum_{|y- x|=1}\a_y q_y}\sum_{i=1}^{N_x}\eta_i^x\s^{(x)}_i\,,\\
H_{N,t}&:=&\sqrt tH_N(\s)+\sqrt {1-t}H'_N(\s)\,.
\eea
We introduce the measures
\be\label{eq:meanvt}
\meanv{\cdot}_{t}:=\frac{\widehat E_{\s^{(1)},\ldots,\s^{(\nu)}}\left[ (\cdot)e^{-\b H_{N,t}}\right]}{ \widehat E_{\s^{(1)}\dots \s^{(\nu)}} [e^{-\b H_{N,t}}]}\,,\quad \meanv{\cdot}_{t=1}=:\meanv{\cdot}_{\b,b,N}\,.
\ee
and
\bea
\mathcal E_N[q]&:=&\sum_{|x-y|=1}\int_0^1dtE_{0,1}E\meanv{\left((R^{(x)}-q_x)(R^{(y)}-q_y)\right)}_{t}\label{eq:error}\\
\RS(q)&:=&\sum_{x=1}^\nu\a_xE_{0,1}\log\cosh\left(b^{(x)}+\b\eta\sqrt{\sum_{|y-x|=1}\a_yq_y}\right)\nn\\
&+&\frac{\b^2}{2}\sum_{|x-y|=1}(\a_x-\a_x q_x)(\a_y-\a_y q_y)\,.\label{eq:ARS}
\eea
Then we have the following sum rule for the free energy
\begin{lemma}\label{lemma:repr}
For every $q=(q_1,\dots, q_\nu)\in[0,1]^\nu$ we have
\be\label{eq:repr}
E[A_N(\b)]+\frac{\b^2}{2}\mathcal E_N[q]=\RS(q)\,.
\ee
\end{lemma}

The proof is quite standard. For the model under consideration it can be found in \cite{min1,min2}, but we give it anyway at the end of the paper for completeness. 

The main contribution of this note is the following
\begin{theorem}\label{TH:minmax}
Assume $|b^{(x)}|\leq \frac12\log\sqrt{2+\sqrt 3}$ for any $x\in[\nu]$. Then
$\RS(q)$ has a unique stationary point $\bar q\in[0,1]^\nu$ satisfying 
\be\label{eq:fixpoints}
\a_{x-1}q_{x-1}+\a_{x+1}q_{x+1}=E_{b^{(x)},\b^2\a_xq_x}[\tanh^2(\eta)]\,,\quad x=1\ldots,[\nu]\,,
\ee
and it holds
\be\label{eq:RSminmax}
\RS(\bar q)=\min_{q_x\,:\,x\in V}\max_{q_x\,:\,x\in H} \RS(q_1,\dots, q_\nu)\,.
\ee
\end{theorem}

This $\min\max$ formula marks an important difference between dRBMs and other better understood spin glass models, such as the Sherrigton-Kirkpatrick model, for which the free energy (and its RS approximation) is given by a convex variational principle. 

In view of (\ref{eq:repr}) the last theorem can be reformulated as follows: there is a $\bar q\in[0,1]^\nu$ independent on $N$ such that 
\be\label{eq:Err-minmax}
\mathcal E_N[\bar q]=\min_{q_x\,:\,x\in V}\max_{q_x\,:\,x\in H} \mathcal E_N[(q_1,\dots, q_\nu)]\,. 
\ee

Theorem \ref{TH:minmax} leads us to the following conjecture
\begin{definition}\label{CJ:main}
Assume there is $\hat q=(\hat q_1,\ldots,\hat q_{\nu})\in[0,1]^\nu$ such that for some choice of the parameters $(\a_1,\ldots,\a_\nu,\b,b^{(1)},\ldots,b^{(\nu)})\in\Omega_{RS}\subseteq(0,1)^{\times \nu}\times(0,\infty)\times \R^\nu$
\be\label{eq:RS-conj1}
\lim_N E\meanv{(\hat q^{(x)}-R^{(x)})^2}_{\b,b,N}=0\,,\quad \forall x\in[\nu]\,.
\ee
Then $\hat q= \bar q$ ($\bar q$ has been introduced in Theorem \ref{TH:minmax}).
\end{definition}
Equivalently we are conjecturing that if the overlaps are self-averaging to some value $\bar q$ then
\be\label{eq:RS-conj2}
\lim_N A_N=\min_{q_x\,:\,x\in V}\max_{q_x\,:\,x\in H} \RS(q_1,\dots, q_\nu)=\RS(\bar q_1,\dots, \bar q_\nu)
\ee
\ie the optimum is attained in $\bar q$. Note that there is at least one case in which the conjecture in this latter form can be a posteriori verified to be true. Indeed if $b^{(x)}=0$ for all $x\in[\nu]$, we can check that $\bar q=0$ satisfies (\ref{eq:fixpoints}). In this case $\RS(0)$ reduces to the annealed free energy and (\ref{eq:RS-conj2}) holds in a certain region of parameters \cite[Theorem 4.1]{min2}; one expects that the self-averaging of the overlaps around zero can be proved in this region by Talagrand's exponential inequalities \cite{Tal}. We also mention the work \cite{min3}, in which the authors recover the free energy (\ref{eq:RSminmax}) on the Nishimori line, which enforces replica symmetry. 

The rest of the paper is organised as follows. We finish the introduction discussing the connection of our work with the existing literature. In Section \ref{sect:GL} we prove a number of auxiliary statements, which basically extend the so-call Guerra-Latala lemma. Then the proof of Theorem \ref{TH:minmax} is given in the Section \ref{sect:proof}. We generalise the optimisation procedure earlier introduced in \cite{NN,bip} to achieve the $\min\max$ formula for the RS free energy of respectively Gaussian-Bernoulli and Bernoulli-Bernoulli RBMs. This optimisation is tricky: after a convenient change of coordinates, we avoid dealing with the Hessian of $\RS(q)$, which turns out to be complicated, and proceed with an iterative nested optimisation of one variable at a time. To do so we will crucially make use of some extensions of the so-called Guerra-Latala Lemma \cite{sum-rules,Tal}. 

\subsection{Related literature}

As already mentioned, a $\min\max$ formula for the RS free energy of RBMs was derived in \cite{NN,bip} for Gauss-Bernoulli and Bernoulli-Bernoulli priors. Interesting enough, for RBMs with spherically symmetric sub-Gaussian priors, albeit their free energy admits also a $\min\max$ formulation \cite{iosfer}, a fully convex variational principle has been proven \cite{ACh, baik, sferico}. Also in problems with similar mathematical settings such as high dimensional linear inference, in which RS is enforced by the Nishimori line, a $\min\max$ formula for the free energy can be achieved by similar methods (see e.g. \cite{macris-rev}) and holds also in the spherical case \cite{macris-sf}. Spherical models are special as they are typically RS and the prior allows diagonalisation of the energy (in the sense of principal values), so that a random matrix type analysis can be performed; we do not known whether this technical advantage is the sole responsible for convexity. 

The mathematical study of multi-layered model initiated with the papers \cite{BCMT} and \cite{pan}. In \cite{BCMT} the authors analysed the Hamiltonian
$$
(\ref{eq:H})+\sum_{x\in[\nu]} \b^{(x)} H^{(x)}_{SK}\,,
$$
where $H^{(x)}_{SK}$ denotes the Sherrington-Kirkpatrick Hamiltonian as a function of the $x$-th layer and $\b^{(x)}\geq0$. They proved an upper bound for the free energy under the assumption that $\b^{(1)},\ldots,\b^{(\nu)}$ are large enough to enforce the positivity of the covariance of the energy in terms of the overlaps. A lower bound was proved in \cite{pan}, independently on this assumption.
These bounds are formulated in terms of a Parisi-like variational principle (therefore convex) and they match for the models studied in \cite{BCMT}. This caused a wide-spread belief that the lower bound of \cite{pan} should be optimal also in the non-convex case, supported by the results for the RS approximation for the two layer case \cite{parisi, bip}, but the question remains open. 

The particular architecture of dRBMs (for Bernoulli priors) has been recently investigated in \cite{min1,min2} in which the authors are able to characterise the high temperature region and to generalise to dRBMs the lower bound of \cite{bip, NN?} for the free energy (notably these are of the same type of the one of \cite{pan}). Furthermore they argue that replica symmetry is identified by a set of equations that we show here to determine the saddle point of $\RS(q)$. In this sense our analysis extends and complements the work begun in \cite{min1,min2}. 

Finally we comment on the RS characterisation given here. A precise mathematical description of the RS phase has been one principal object of study at the early stage of the theory of disordered system, especially for the Sherrington-Kirkpatrick model in the attempt to disprove the replica symmetric formula of the free energy. The paradigm was to prove the following statement: {\em if the overlap is self-averaging then the free energy is equal to its RS expression}. To prove that, the main point seems to be to select the right sequence of finite dimensional Gibbs measure w.r.t. which the overlap should be asymptotically self-averaging. In \cite{pash} the authors proved the statement working with a Gibbs measure biased by a very specific random vector vanishing as $N^{-\frac14}$. This result was later revisited in \cite{sh} and in \cite{talPTRF}: in either paper it is considered a Gibbs measure perturbed according to the so-called cavity method; the latest version of this proof can be found in \cite[Section 1.6]{Tal}. 
Finally from the approach of \cite{sum-rules} it can be easily proven that if the interpolation error, \ie the variance of the overlap w.r.t. the interpolating Gibbs measure, vanishes in some point, then the $\min$ of the RS functional must be attained only in its zero. This allows us to conclude that such a self-averaging property of the overlap implies the RS formula for the free energy. 
Once again for dRBMs the crucial difference is in their non-convex structure, which complicates things further, as the interpolation error is not positive definite. 

\subsection{Acknowledgements} This paper benefited greatly from the observations of an anonymous referee, who is gratefully acknowledged.

\section{Guerra-Latala lemmas}\label{sect:GL}

The proof of our main theorem will heavily use the so-called Guerra-Latala lemma \cite{sum-rules}\cite{Tal}. In the next two lemmas we follow essentially \cite[A.14]{Tal}, albeit with a different presentation, obtaining somewhat more general results.

We begin by the following formula, holding for instance for any $f\in C^2(\R)$:
\be\label{eq:derivata-sigma}
\frac{d}{d\s}E_{b,\s}[f]=\frac{E_{b,\s}[(\eta-b)f'(\eta)]}{2\s}\,. 
\ee
The proof is simple:
\be\label{eq:proof-der-sigma}
\frac{d}{d\s}E_{b,\s}[f(\eta)]=\frac{d}{d\s}E_{0,1}[f(\eta\sqrt\s+b)]
=\frac{1}{2\sqrt\s}E_{0,1}[\eta f'(\eta\sqrt\s+b)]=\frac{E_{b,\s}[(\eta-b)f'(\eta)]}{2\s}\,.
\ee
Let us now assume $\Phi\in C^2(\R)$, odd, strictly increasing, with $E_{0,1}[|\Phi|]$ bounded, $\Phi(0)=0$ and $x\Phi''(x)<0$ for $x\neq0$. We also assume that there is a number $L>0$ such that $\Phi^2$ is convex for $|x|\leq L$. Later on we will specialise $\Phi(x)=\tanh(x)$, which clearly enjoys these properties (with $L=\log\sqrt{2+\sqrt 3}$). 

We set for any $b\in\R$, $\s>0$
\be\label{eq:c}
c_{b,\s}[\Phi]:=\frac{E_{b,\s}[(\eta-b)\Phi(\eta)\Phi'(\eta)]}{E_{b,\s}[\Phi^2]}\,. 
\ee

\begin{lemma}\label{lemma:c}
Let $\Phi$ and $L>0$ defined as above. Then $c_{b,\s}[\Phi]\in(0,1)$ for all $|b|\leq L/2$.
\end{lemma}
\begin{proof}
First we show the upper bound, in fact without any restriction on $b$. We use that $\Phi(z)\Phi'(z)$ is odd to have
\bea
bE_{b,\s}[\Phi(\eta)\Phi'(\eta)]&=&\frac b2 \left(E_{b,\s}[\Phi(\eta)\Phi'(\eta)]-E_{-b,\s}[\Phi(\eta)\Phi'(\eta)]\right)\nn\\
&=&be^{-\frac{b^2}{2\s}}E_{b,\s}[\sinh (\frac{b}{\s}\eta)\Phi(\eta)\Phi'(\eta)]\geq0\,,\nn
\eea
uniformly in $b$. Then
$$
c_{b,\s}[\Phi]\leq\frac{E_{b,\s}[\eta\Phi(\eta)\Phi'(\eta)]}{E_{b,\s}[\Phi^2]}\,. 
$$
Now we show that 
$$
E_{b,\s}[\Phi^2]>E_{b,\s}[\eta\Phi(\eta)\Phi'(\eta)]\,. 
$$
We study the function $x\mapsto x\Phi'(x)-\Phi(x)$ for $x\in\R$. We note that it vanishes at the origin and it decreases for $x\neq0$. So it must be $x\Phi'(x)-\Phi(x)>0$ for $x<0$ and $x\Phi'(x)-\Phi(x)<0$ for $x>0$, whence
$$
\Phi(x)(x\Phi'(x)-\Phi(x))<0\,,\quad \forall \,\,x\neq0\,. 
$$

For the lower bound on $c_{b,\s}[\Phi]$ we consider for definiteness $b>0$. As $\Phi(\eta)\Phi'(\eta)$ is an odd function then $\eta\Phi(\eta)\Phi'(\eta)$ is even, thus we have
\bea
E_{b,\s}[(\eta-b)\Phi(\eta)\Phi'(\eta)]&=&e^{-\frac{b^2}{2\s}}E_{0,\s}[(\eta-b)\Phi(\eta)\Phi'(\eta)e^{\frac{b\eta}{\s}}]\nn\\
&=&\frac{e^{-\frac{b^2}{2\s}}}{2}E_{0,\s}[(\eta-b)\Phi(\eta)\Phi'(\eta)e^{\frac{b\eta}{\s}}]+\frac{e^{-\frac{b^2}{2\s}}}{2}E_{0,\s}[(-\eta-b)\Phi(-\eta)\Phi'(-\eta)e^{-\frac{b\eta}{\s}}]\nn\\
&=&e^{-\frac{b^2}{2\s}}E_{0,\s}[\Phi(\eta)\Phi'(\eta)(\eta\cosh(\frac{\eta b}{\s}))-b\sinh(\frac{\eta b}{\s}))]\nn\\
&=&2e^{-\frac{b^2}{2\s}}E_{0,\s}[1_{\{\eta\geq0\}}\Phi(\eta)\Phi'(\eta)\cosh(\frac{\eta b}{\s})(\eta-b\tanh(\frac{\eta b}{\s}))]\nn\,\\
&=&\sum_{s=\pm1}E_{sb,\s}[1_{\{\eta\geq0\}}\Phi(\eta)\Phi'(\eta)(\eta-b\tanh(\frac{\eta b}{\s}))]\label{eq:above}\,.
\eea
Let $\eta^*$ denote the unique non-zero solution of $\eta=b\tanh(\frac{\eta b}{\s})$. It is easy to see that the origin is the unique solution of such equation iff $\s\geq b^2$. Moreover if $\s\leq b^2$ the function $\eta^*(\s)\geq0$ is decreasing (as can be seen by implicit differentiation) from $\eta^*(0)=b$ to zero. 
Therefore the integrand in \eqref{eq:above} is negative for $\eta<\eta^*$, otherwise it is positive. So if $\s\geq b^2$ then $\eta\geq b\tanh(\frac{\eta b}{\sqrt \s})$ and the whole expression \eqref{eq:above} is positive. 

Assuming now $\s\in[0,b^2]$, we split
\bea
E_{sb,\s}[1_{\{\eta\geq0\}}\Phi(\eta)\Phi'(\eta)(\eta-b\tanh(\frac{\eta b}{\s}))]&=&-E_{sb,\s}[1_{\{0\leq\eta<\eta^*\}}\Phi(\eta)\Phi'(\eta)|\eta-b\tanh(\frac{\eta b}{\s})|]\nn\\
&+&E_{sb,\s}[1_{\{\eta\geq\eta^*\}}\Phi(\eta)\Phi'(\eta)|\eta-b\tanh(\frac{\eta b}{\s})|]\,. 
\eea

We will use the estimate
$$
b\tanh(\frac{\eta b}{\s})\leq \eta^*+\frac{b^2-{\eta^*}^2}{\s}(\eta-\eta^*)
$$
coming from Taylor expansion. We have
\bea
&&E_{sb,\s}[1_{\{0\leq\eta<\eta^*\}}\Phi(\eta)\Phi'(\eta)|\eta-b\tanh(\frac{\eta b}{\s})|]\nn\\
&\leq&\left(1-\frac{b^2-{\eta^*}^2}{\s}\right)\Phi(\eta^*)\Phi'(\eta^*)E_{sb,\s}[1_{\{0\leq\eta<\eta^*\}}(\eta^*-\eta)]\label{eq:BBB1}
\eea
and
\bea
&&E_{sb,\s}[1_{\{\eta\geq\eta^*\}}\Phi(\eta)\Phi'(\eta)|\eta-b\tanh(\frac{\eta b}{\s})|]\nn\\
&\geq&E_{sb,\s}[1_{\{\eta^*\leq\eta\leq 2b\}}\Phi(\eta)\Phi'(\eta)|\eta-b\tanh(\frac{\eta b}{\s})|]\nn\\
&\geq&\left(1-\frac{b^2-{\eta^*}^2}{\s}\right)\Phi(\eta^*)\Phi'(\eta^*)E_{sb,\s}[1_{\{\eta^*\leq\eta\leq 2b\}}(\eta-\eta^*)]\,.\label{eq:BBB2}
\eea
We used in (\ref{eq:BBB1}) ansd \eqref{eq:BBB2} that $\Phi\Phi'$ is positive and increasing for $x\in[0,2b]$ (since $b<L/2$); thus
\bea
&&E_{b,\s}[1_{\{\eta\geq0\}}\Phi(\eta)\Phi'(\eta)(\eta-b\tanh(\frac{\eta b}{\s}))]\nn\\
&\geq& \left(1-\frac{b^2-{\eta^*}^2}{\s}\right)\Phi(\eta^*)\Phi'(\eta^*)E_{b,\s}[1_{\{0\leq\eta\leq 2b\}}(\eta-\eta^*)]\nn\\
&=& \left(1-\frac{b^2-{\eta^*}^2}{\s}\right)\Phi(\eta^*)\Phi'(\eta^*)\left(E_{b,\s}[1_{\{0\leq\eta\leq 2b\}}(\eta-b)]+(b-\eta^*)E_{b,\s}[1_{\{0\leq\eta\leq 2b\}}]\right)\nn\\
&=& \left(1-\frac{b^2-{\eta^*}^2}{\s}\right)\Phi(\eta^*)\Phi'(\eta^*)(b-\eta^*)E_{b,\s}[1_{\{0\leq\eta\leq 2b\}}]>0\label{eq:maggioredizero}\,.
\eea
On the other hand, neglecting completely the positive part, we have by (\ref{eq:BBB1})
$$
E_{-b,\s}[1_{\{\eta\geq0\}}\Phi(\eta)\Phi'(\eta)(\eta-b\tanh(\frac{\eta b}{\s}))]\geq -\left(1-\frac{b^2-{\eta^*}^2}{\s}\right)\Phi(\eta^*)\Phi'(\eta^*)\sqrt2e^{-\frac{b^2}{4\s}}\eta^*\,. 
$$
Combining with \eqref{eq:above} and \eqref{eq:maggioredizero} we get
\bea
&&\sum_{s=\pm1}E_{sb,\s}[1_{\{\eta\geq0\}}\Phi(\eta)\Phi'(\eta)(\eta-b\tanh(\frac{\eta b}{\s}))]\nn\\
&\geq& \left(1-\frac{b^2-{\eta^*}^2}{\s}\right)\Phi(\eta^*)\Phi'(\eta^*)\left((b-\eta^*)E_{b,\s}[1_{\{0\leq\eta\leq 2b\}}] -\sqrt2e^{-\frac{b^2}{4\s}}\eta^*\right)\nn\\
&=& \!\!\!\left(1-\frac{b^2-{\eta^*}^2}{\s}\right)\Phi(\eta^*)\Phi'(\eta^*)\left((b-\eta^*)\left(E_{0,1}[1_{\{|\eta|\leq \frac{b}{\sqrt\s}\}}] +\sqrt2e^{-\frac{b^2}{4\s}}\right)-b\sqrt2e^{-\frac{b^2}{4\s}}\right).\,\,\label{eq:expressionaboveultima}
\eea
Noting that $b-\eta^*=b(1+e^{\frac{2b\eta^*}{\s}})^{-1}$ it is immediate to verify that the expression \eqref{eq:expressionaboveultima} is non-negative for all $\s\in[0,b^2]$, which ends the proof.
\end{proof}

For any $f\in C^2(\R^+)$ positive and increasing, by (\ref{eq:derivata-sigma}), (\ref{eq:c}) and Lemma \ref{lemma:c} we have 
\be\label{eq:Phi-increasing}
\frac{d}{dx}E_{b,f(x)}[\Phi^2] = f'(x)\frac{E_{b,f(x)}[(\eta-b)\Phi(\eta)\Phi'(\eta)]}{f(x)}=f'(x)\frac{c_{b,f(x)}[\Phi]E_{b,f(x)}[\Phi^2]}{f(x)}>0\,,
\ee
\ie $E_{b,f(x)}[\Phi^2]$ is increasing w.r.t. $x$ for any $|b|\leq L/2$. We will use much this fact in the sequel. 

\begin{lemma}[Guerra-Latala]\label{lemma:GLur}
Let $b\in\R$ and $f,g$ be differentiable, positive and increasing.
The function $\Psi\,:\,\R^+\mapsto \R^+$ defined by
\be\label{eq:GL-psi}
\Psi(x):=\frac{E_{b,f(x)}[\Phi^2]}{g(x)}
\ee
is decreasing for those $x$ such that
\be\label{eq:cond-fg}
g'(x)f(x)\geq c_{b,f(x)}[\Phi]f'(x)g(x)\,. 
\ee
Otherwise it is increasing. 
\end{lemma}

\begin{proof}
The assertion follows readily from 
\be\label{eq:GLmain}
\Psi'(x)=\Psi(x)\left(\frac{f'(x)c_{b,f(x)}[\Phi]}{f(x)}-\frac{g'(x)}{g(x)}\right)\,.
\ee
To prove it, we compute by \eqref{eq:derivata-sigma}
$$
(g\Psi)'=\frac{f'}{f}\E_{b,f(x)}[(\eta-b)\Phi(\eta)\Phi'(\eta)]\,. 
$$
Therefore by the Leibniz rule
\bea
g(x)\Psi'(x)&=&\frac{f'(x)}{f(x)}\E_{b,f(x)}[(\eta-b)\Phi(\eta)\Phi'(\eta)]-\frac{g'(x)}{g(x)}E_{b,f(x)}[\Phi^2]\nn\\
&=&E_{b,f(x)}[\Phi^2]\left(\frac{f'(x)c_{b,f(x)}[\Phi]}{f(x)}-\frac{g'(x)}{g(x)}\right)\,,
\eea
which proves (\ref{eq:GLmain}). 
\end{proof}

The usual Guerra-Latala result corresponds to the special case $f(x)=g(x)=x$. In turn the simpler bound 
\be\label{eq:simple-bound}
g'f\geq f'g\,
\ee
will be often sufficient to make $\Psi$ decreasing thanks to Lemma \ref{lemma:c}. This amounts to assume $f/g$ non-increasing. We state for clarity this simplified version of the above lemma as follows:

\begin{lemma}\label{lemma:GL}
Let $|b|\leq L/2$ and $f,g$ be differentiable, positive and increasing such that (\ref{eq:simple-bound}) holds. 
Then the function $\Psi$ defined by (\ref{eq:GL-psi}) is decreasing.
\end{lemma}

The following cases are not covered by the previous lemmas.
\begin{lemma}\label{lemma:utile0}
Let $G\,:\,\R^+\mapsto\R$ be continuously differentiable and increasing. Assume $|b|\leq L/2$ and that there is a unique $(x_0,y_0)$ such that
\be\label{eq-dini}
G(x_0)=E_{b,y_0}[\Phi^2]\,. 
\ee
Then there is a function $\bar y$ uniquely defined by
\be\label{eq:utile1bis}
G(x)-E_{b,\bar y(x)}[\Phi^2]=0\,.
\ee
Moreover $\bar y$ and $\bar y/G$ are increasing.
\end{lemma}
\begin{proof}
$G(x)$ is increasing and $-E_{b,y}[\Phi^2]$ is decreasing. Therefore the implicit function theorem applies, giving the existence of a differentiable function $\bar y(x)>0$, uniquely defined in a neighbourhood of $(x_0, y_0)$ by
\be\label{eq;via}
G(x)=E_{b,\bar y(x)}[\Phi^2]\,.
\ee
Then $\bar y(x)$ extends globally as a function $\R^+\mapsto \R^+$ via (\ref{eq;via}) since the intersection point $(x_0, y_0)$ is unique. We compute the derivative via
$$
\bar y'(x)=\frac{yG'(x)}{E_{b,y}[(\eta-b)\Phi'\Phi]}\Big|_{y=\bar y}=\frac{\bar yG'(x)}{c_{b,\bar y}[\Phi]G(x)}>0\,,
$$
where we used (\ref{eq:Phi-increasing}). Thus $\bar y$ is differentiable, positive and increasing as function of $x$. Then, considering
\be
\frac{G(x)}{\bar y(x)}=\frac{E_{b,\bar y}[\Phi^2]}{\bar y(x)}
\ee
we see that the r.h.s. is decreasing by Lemma \ref{lemma:GL} (with $f=g=\bar y$). Thereby $\bar y(x)/G(x)$ is increasing.
\end{proof}

\begin{lemma}\label{lemma:utile1bis}
Let $F\,:\,\R^+\mapsto\R$ be non-negative, continuously differentiable, with $F'(y)y\geq F(y)$. Assume $|b|\leq L/2$ and that there is a unique $(x_0,y_0)$ such that
$$
F(y_0)=E_{b,y_0+x_0}[\Phi^2]\,. 
$$
Then there is a unique function $\bar y$ defined by
\be\label{eq:utile1bis}
F(y)-E_{b,y+x}[\Phi^2]=0\,,\quad x\geq0\,.
\ee
The function $\bar y$ is increasing and $\bar y(x)/x$ is decreasing. Moroever
\be\label{eq:moreover}
\partial_y(F(y)-E_{b,y+x}[\Phi^2])\Big |_{y=\bar y}>0\,.
\ee
\end{lemma}

\begin{proof}
By (\ref{eq:Phi-increasing}) $E_{b,y+x}[\Phi^2]$ is increasing in $x$ at fixed $y$ and in $y$ at fixed $x$. By Lemma \ref{lemma:GL} (with $f=y+x$ and $g=y$) we see that 
$$
\frac{E_{b,y+x}[\Phi^2]}{y}
$$
is decreasing as a function of $y$. Therefore $E_{b,y+x}[\Phi^2]$ is increasing and concave, while by assumption $F(y)$ is increasing and convex. Thus the function
$$
\frac{F(y)}{y}-\frac{E_{b,y+x}[\Phi^2]}{y}
$$
is increasing. Therefore
\be
\partial_y(F(y)-E_{b,y+x}[\Phi^2])\Big|_{(x,y)=(x_0,y_0)}>\frac{F(y)}{y}-\frac{E_{b,y+x}[\Phi^2]}{y}\Big|_{(x,y)=(x_0,y_0)}=0\,.
\ee
Hence by the implicit function theorem we can define a function $\bar y$ locally around $x_0$ such that $\bar y(x_0)=y_0$ and (\ref{eq:moreover}) holds.
Combining with \eqref{eq:Phi-increasing} we get
\be
\bar y'(x)=\frac{\partial_x E_{b,y+x}[\Phi^2]}{\partial_y(F(y)-E_{b,y+x}[\Phi^2])}\Big |_{y=\bar y}>0\,,
\ee
so $\bar y$ is increasing in its domain. Moreover computing explicitly by \eqref{eq:derivata-sigma} and using Lemma \ref{lemma:c} yields
\bea
\bar y'(x)&=&\frac{c_{b,\bar y+x}[\Phi]F(\bar y)}{F'(\bar y)x+F'(\bar y)y-c_{b,\bar y+x}[\Phi]F(\bar y)}\nn\\
&=&\frac{c_{b,\bar y+x}[\Phi]F(\bar y)}{F'(\bar y)x+(1-c_{b,y+x}[\Phi])F(\bar y)}\nn\\
&<& \frac{F(\bar y)}{F'(\bar y)x} \leq \frac{ \bar y}x\,,\label{eq:y/x,decreasing}
\eea
which implies $\bar y/x$ decreasing in its domain. 
\end{proof}

\begin{lemma}\label{lemma:GLsumma}
Assume $|b|\leq L/2$. Let $F\,:\,\R^+\mapsto\R$ be non-negative, continuously differentiable, with $F'(y)y\geq F(y)$ and $y$ implicitly defined by
\be\label{eq:utile1bis}
F(y)=E_{b,y+x}[\Phi^2]\,,\quad x\geq0\,.
\ee
Then 
\be\label{eq:GL-psi1}
\Psi_1(x):=\frac{E_{b,y+x}[\Phi^2]}{x}
\ee
is decreasing.
\end{lemma}

\begin{proof}
We apply Lemma \ref{lemma:GL}  with $f=y+x$ and $g=x$ and we see that condition (\ref{eq:simple-bound}) implying $\Psi_1$ decreasing reads $xy'(x)\leq y(x)+x$. This is ensured for all $x>0$ by Lemma \ref{lemma:utile1bis}. 
\end{proof}


\section{Proofs of Theorem \ref{TH:minmax} and Lemma \ref{lemma:repr}}\label{sect:proof}

Here we prove Theorem \ref{TH:minmax}.
We rewrite the RS function we want to optimise
\be\label{eq:RS-again}
\RS(q):=\sum_{x=1}^\nu\a_xE_{0,1}\left[\log\cosh \left(b^{(x)}+\b\eta\sqrt{\sum_{|y-x|=1}\a_y q_y}\right)\right]+\frac{\b^2}{2}\sum_{|x-y|=1}(\a_x-\a_x q_x)(\a_y-\a_y q_y)\,.
\ee
Recall that we assumed the elements of $V$ to be the first $|V|$ odd numbers in $\{1,\ldots \nu\}$.
Note that if $x\in V$ then $\{y\in[\nu]\,:\,|x-y|=1\}\subseteq H$.
We set for $x\in V$
\be\label{eq:q->Q}
Q_x:=\sum_{|y-x|=1} \a_y q_y\,,\quad A_x:=\sum_{|y-x|=1}\a_y\,.
\ee
Note that $Q_x\leq A_x$. We shorten also $p_x:=\a_xq_x$ for $x\in V$.
Next we change coordinates and describe $\RS(q)$ only in the variables $\{ Q_x, p_x\}_{x\in V}$. Then (\ref{eq:RS-again}) is rewritten as
\bea
\overline{\RS}(p,Q)&=&\sum_{x\in V}\a_x E_{0,1}\left[\log\cosh\left(b^{(x)}+\b\eta\sqrt{Q_x}\right)\right]\nn\\
&+&\sum_{x\in H}\a_x E_{0,1}\left[\log\cosh\left(b^{(x)}+\b\eta\sqrt{\sum_{|y-x|=1}p_y}\right)\right]\nn\\
&+&\frac{\b^2}{2}\sum_{x\in V}(\a_x-p_x)(A_x-Q_x)\,.\label{eq:def-RS}
\eea

With a small abuse of notation we will denote the different functions $\overline\RS(q,Q), \overline\RS(p,Q)$ with the same symbol; the meaning will always be clear by the context. We will prove that there is a unique stationary point $(\bar q,\bar Q)\in[0,1]^{|V|}\times[0,1]^{|V|}$ in which the $\min\max$ is realised:
\be\label{eq:toprove}
\min_{q_{x_1},\ldots,q_{x_{|V|}}}\max_{Q_{x_1},\ldots,Q_{x_{|V|}}} \overline\RS(q,Q)=\overline\RS(\bar q,\bar Q)\,. 
\ee
Then, since the change of variables (\ref{eq:q->Q}) is linear and injective, Theorem \ref{TH:minmax} follows. 

In the proof below we shall use the lemmas of Section \ref{sect:GL} always with $\Phi=\tanh$ and assuming that for all $x\in[\nu]$ it is $|b^{(x)}|\leq \frac12\log\sqrt{2+\sqrt 3}$.

\begin{proof}[Proof of \eqref{eq:toprove}] 

We will prove that the stationary points equations
\be
\partial_{Q_x}\overline \RS=0\,, \partial_{p_x}\overline \RS=0\,,\quad x=1,\ldots,\nu\,,
\ee
have a unique solution in which the minimax is attained. 

We first take the derivatives w.r.t. $Q_{x}$ for all $x\in V$. We have
\be\label{eq:der_p1}
\partial_{Q_x}\overline \RS=\frac{\b^2}{2}\left(p_x-\a_xE_{b^{(x)},\b^2 Q_x}\left[\tanh^2\left(\eta\right)\right]\right)\,,\quad \forall\,\,x\in V\,, 
\ee
where we used a Gaussian integration by parts. 

By (\ref{eq:Phi-increasing}) the function $Q_x\mapsto\a_xE_{b^{(x)},\b^2 Q_x}\left[\tanh^2\left(\eta\right)\right]$ is increasing in $Q_x$ at fixed $p_x$. Hence there is a unique point $(p_{x,0}, Q_{x,0})$ such that \eqref{eq:der_p1} vanishes. Then by Lemma \ref{lemma:utile0} (with $G(x)=x$, $b=b^{(x)}$) there is a unique function $Q_x(p_x)$ defined in a neighbourhood of the intersection point.
We shorten
\be\label{eq:Cx}
B_x:=\a_x\tanh^2(b^{(x)})\,\quad C_x:=\a_xE_{b^{(x)},\b^2A_x}\left[\tanh^2\left(\eta\right)\right]\,.
\ee
So the stationary point conditions (\ref{eq:der_p1}) define uniquely the functions 
\be\label{eq:selfQ0}
Q_x\,:\,\left[B_x,C_x\right]\mapsto[0,A_x] 
\ee
via 
\be\label{eq_selfQ}
p_x=\a_xE_{b^{(x)}, \b^2Q_x(p_x)}\left[\tanh^2\left(\eta\right)\right]\,,\quad x\in V\,.
\ee
The functions $Q_x(p_x)$ are non-negative, with $Q_x(B_x)=0$, $Q_x(C_x)=A_x$ and we can prolong them in the interval $[0,C_x]$ setting $Q_x=0$ in $[0,B_x]$. Moreover again by Lemma \ref{lemma:utile0} they are increasing and convex with
\be\label{eq:Q'}
\frac{d}{dp_x}\frac{Q_x(p_x)}{p_x}>0\,.
\ee
Therefore
$$
\partial^2_{Q_x}\overline{RS}<0\,,\forall x\in V\,,
$$
and clearly
$$
\partial^2_{Q_xQ_y}\overline{RS}=0\,,\forall x\neq y\,. 
$$

Now we set
\be\label{eq:ARSridotta1}
\overline\RS^{(1)}(p):=\max_{Q_{1}\ldots,Q_{{|V|}}} \overline \RS(p,Q)=\overline\RS(Q_{1}(p_{1}),Q_{3}(p_{3})\ldots,Q_{2|V|-1}(p_{2|V|-1}), p_{1},\ldots,p_{2|V|-1})\,.
\ee

We compute
\be\label{eq:der-p1}
\partial_{p_1}\overline\RS^{(1)}=\frac{\b^2}{2}\left(Q_1( p_1)-\a_2E_{b^{(2)},\b^2(p_1+p_3)}\left[\tanh^2\left(\eta\right)\right]\right)\,.
\ee
First we study the stationary point equation
\be\label{eq:stat1}
Q_1( p_1)-\a_2E_{b^{(2)},\b^2(p_1+p_3)}\left[\tanh^2\left(\eta\right)\right]=0\,.
\ee
We look at
\be\label{eq:z1}
z(p_1):=\frac{\a_2E_{b^{(2)},\b^2(p_1+\a_3)}\left[\tanh^2\left(\eta\right)\right]}{Q_1( p_1)}\,.
\ee
Since $\frac{Q_1(p_1)}{p_1}$ is monotone increasing (see (\ref{eq:Q'})) and by Lemma \ref{lemma:GL} (with $f(p_1)=p_1+\a_3$ and $g(p_1)=Q_1(p_1)$) 
$$
\frac{\a_2E_{b^{(2)},\b^2(p_1+\a_3)}\left[\tanh^2\left(\eta\right)\right]}{p_1}
$$
is monotone decreasing, we have that $z(p_1)$ is decreasing from $\lim_{p_1\to B_1}z(p_1)=\infty$ to $z(C_1)$. 
Therefore it exists a unique solution $(p_{3,0},p_{1,0})$ to (\ref{eq:stat1}) 
provided 
\be\label{eq:cond1}
z_1(C_1)\leq1\,,
\ee
which is always satisfied (since $Q_1(C_1)=A_1=\a_2$). 

As $\frac{Q_1(p_1)}{p_1}$ is monotone increasing (see (\ref{eq:Q'})) we can apply Lemma \ref{lemma:utile1bis} (with $F=Q_1$, $(x,y)=(p_{3,0},p_{1,0})$). This yields that the stationary point equation (\ref{eq:stat1})
defines uniquely a function $p_1(p_3)$ which satisfies
\be\label{eq:der-p1-IFT}
\frac{d}{dp_3}p_1(p_3)>0\,,\qquad \frac{d}{dp_3}\frac{p_1(p_3)}{p_3}<0\,.
\ee
Moreover (see (\ref{eq:moreover}))
\be\label{eq:derIIRS1}
\partial_{p_1}\left(Q_1( p_1)-\a_2E_{b^{(2)},\b^2(p_1+p_3)}\left[\tanh^2\left(\eta\right)\right]\right)\Big|_{p_1=p_1(p_3)}>0\,,
\ee

Finally by Lemma \ref{lemma:GLsumma} we get
\be\label{eq:p1(p3)}
\frac{d}{dp_3}\frac{E_{b^{(2)},\b^2(p_1+p_3)}\left[\tanh^2\left(\eta\right)\right]}{p_3}<0\,. 
\ee

The inequality \eqref{eq:derIIRS1} is equivalent to $\partial^2_{p_1}\overline\RS^{(1)}\Big|_{p_1=p_1(p_3)}>0$ which ensures that $p_1(p_3)$ is a minimum point. So we can set
\be
\overline{RS}^{(3)}(p_3,\ldots,p_{|V|}):=\min_{p_1} \overline{RS}^{(1)}(p_1,\ldots,p_{|V|})= \overline{RS}^{(1)}(p_1(p_3),p_3,\ldots,p_{|V|})\,. 
\ee

\

Now we compute
\be\label{eq:der-p2}
\partial_{p_3}\overline{RS}^{(3)}=\frac{\b^2}{2}\left(Q_3( p_3)-\a_2E_{b^{(2)}, \b^2(p_1(p_3)+p_3)}\left[\tanh^2(\eta)\right]-\a_4E_{b^{(2)}, \b^2(p_3+p_5)}\left[\tanh^2(\eta)\right]\right)\,
\ee
So the stationary point equation reads as
\be\label{eq:stat2}
Q_3(p_3)-\a_2E_{b^{(2)}, \b^2(p_1(p_3)+p_3)}\left[\tanh^2(\eta)\right]-\a_4E_{b^{(2)}, \b^2(p_3+p_5)}\left[\tanh^2(\eta)\right]=0\,.
\ee
We set
\be\label{eq:z3}
z_3(p_3):=\frac{\a_2E_{b^{(2)}, \b^2(p_1(p_3)+p_3)}\left[\tanh^2(\eta)\right]}{Q_3( p_3)}+\frac{\a_4E_{b^{(2)},\b^2(p_3+\a_5)}\left[\tanh^2(\eta)\right]}{Q_3( p_3)}\,. 
\ee
By (\ref{eq:Q'}) and by the first inequality of (\ref{eq:p1(p3)}) we have
$$
\frac{Q'_3(p_3)}{Q_3(p_3)}\geq \frac{1}{p_3}\geq \frac{p'_1(p_3)}{p_1(p_3)}\,.
$$
Then applying Lemma \ref{lemma:GL} for each summand, that is with $f=p_3+\a_5$ or $f=p_1(p_3)+p_3$ and $g=Q_3(p_3)$ we have that $z_3(p_3)$ is decreasing. Therefore (\ref{eq:stat2}) yields a unique intersection point $(p_{5,0},p_{3,0})$ provided
\be\label{eq:cond2}
z_3(C_3)\leq1\,,
\ee
which holds true for all the values of the parameters. 

Let us now set
$$
F_3(p_3):=Q_3( p_3)-\a_2E_{b^{(2)}, \b^2(p_1(p_3)+p_3)}\left[\tanh^2(\eta)\right]\,.
$$
Note that $F_3(p_3)/p_3$ is monotone increasing: this follows by \eqref{eq:Q'} and \eqref{eq:p1(p3)}. Therefore by Lemma \ref{lemma:utile1bis} (with $F=F_3$ and $(x,y)=(p_{5,0},p_{3,0})$) the stationary point equation (\ref{eq:stat2}) singles out a unique function $p_3(p_5)$ which satisfies
\be\label{eq:der-p1-IFT}
\frac{d}{dp_5}p_3(p_5)>0\,,\qquad \frac{d}{dp_5}\frac{p_3(p_5)}{p_5}<0\,.
\ee
Therefore $p_3(p_5)$ is increasing. 
Moreover (see (\ref{eq:moreover}))
\be\label{eq:derIIRS3}
\partial_{p_3}\left(F_3( p_3)-\a_4E_{b^{(2)}, \b^2(p_3+p_5)}\left[\tanh^2(\eta)\right]\right)\Big|_{p_3=p_3(p_5)}>0\,,
\ee

Finally by Lemma \ref{lemma:GLsumma} we get
\be\label{eq:p3(p5)}
\frac{d}{dp_5}\frac{E_{b^{(2)},\b^2(p_3(p_5)+p_5)}\left[\tanh^2(\eta)\right]}{p_5}<0\,. 
\ee

The inequality (\ref{eq:derIIRS3}) is equivalent to $\partial^2_{p_2}\overline\RS^{(3)}\Big|_{p_3=p_3(p_5)}>0$, that is in $p_3(p_5)$ a minimum is attained. So we can set
\be
\overline{RS}^{(5)}(p_5,\ldots,p_{|V|}):=\min_{p_3} \overline{RS}^{(3)}(p_1(p_3),p_3,\ldots,p_{|V|})= \overline{RS}^{(3)}(p_1(p_3(p_5)),p_3(p_5),p_5,\ldots,p_{|V|})\,. 
\ee

\

In general, for $x\in(2\N-1)\cap[\nu]\setminus \{1\}$ at $x$-th step we have already proven
\be\label{eq:aux-x-1}
\frac{d}{dp_{x}}p_{x-2}(p_{x})>0\,,\quad\frac{d}{dp_{x}}\frac{p_{x-2}(p_{x})}{p_{x}}<0\,,\quad \frac{d}{dp_{x}}\frac{E_{b^{(x-1)},\b^2(p_{x-2}(p_{x})+p_{x})}\left[\tanh^2(\eta)\right]}{p_{x}}<0\,.
\ee
We shorten in what follows $p_{x-2}(p_x)=p_{x-2}$. We want to optimise only w.r.t. $p_x$ the function $\overline\RS^{(x-1)}=\overline\RS^{(x-1)}(p_{x},\ldots,p_{|V|})$ defined by 
\be\label{eq:ARSridottak}
\overline\RS^{(x)}:=\min_{p_{x-2}}\overline\RS^{(x-2)}\,.
\ee


The stationary point equation reads as
\be\label{eq:statx}
Q_x(p_x)-\a_{x-1}E_{b^{(x-1)},\b^2(p_{x-2}(p_x)+p_x)}\left[\tanh^2(\eta)\right]-\a_{x+1}E_{b^{x+1},\b^2(p_x+p_{x+2})}\left[\tanh^2(\eta)\right]=0\,.
\ee
We set
\be\label{eq:zx}
z_x(p_x):=\frac{\a_{x-1}E_{b^{(x-1)},\b^2(p_{x-2}(p_x)+p_x)}\left[\tanh^2(\eta)\right]+\a_{x+1}E_{b^{(x+1)},\b^2(p_{x-2}(p_x)+\a_{x+2})}\left[\tanh^2\left(\eta\right)\right]}{Q_x( p_x)}\,. 
\ee
By (\ref{eq:Q'}) and by the second inequality of \eqref{eq:aux-x-1} we have
$$
\frac{Q'_x(p_x)}{Q_x(p_x)}\geq \frac{1}{p_x}\geq \frac{p'_{x-2}(p_x)}{p_{x-2}(p_x)}\,.
$$
Then applying Lemma \ref{lemma:GL} for each summand, that is with $f=p_x+\a_{x+2}$ or $f=p_{x-2}(p_{x})+p_x$ and $g=Q_x(p_x)$ we have that $z_x(p_x)$ is decreasing. Therefore (\ref{eq:statx}) yields a unique intersection point $(p_{x,0},p_{x+2,0})$ provided
\be\label{eq:condx}
z_x(C_x)\leq1\,,
\ee
which is always satisfied.


We set now
\be
F_x(p_x):=Q_x( p_x)-\a_{x-1}E_{b^{(x-1)},\b^2(p_{x-2}+p_{x})}\left[\tanh^2(\eta)\right]
\ee
and see that $F_x(p_x)/p_x$ is increasing thanks to the combination of (\ref{eq:Q'}) and \eqref{eq:aux-x-1}. 
Therefore we can apply Lemma \ref{lemma:utile1bis} (with $F=F_x$ and $(x,y)=(p_{x+2,0},p_{x,0})$) to prove that the stationary point equation (\ref{eq:statx}) singles out a unique function $p_x(p_{x+2})$ with
\be\label{eq:der-p1-IFT}
\frac{d}{dp_{x+2}}p_{x}(p_{x+2})>0\,,\qquad \frac{d}{dp_{x+2}}\frac{p_x(p_{x+2})}{p_{x+2}}<0\,. 
\ee
Moreover (see (\ref{eq:moreover}))
\be\label{eq:derIIRSx}
\partial_{p_x}\left(F_x( p_x)-\a_{x+1}E_{b^{(x+1)},\b^2(p_x+p_{x+2})}\left[\tanh^2(\eta)\right]\right)\Big|_{p_x=p_x(p_{x+2})}>0\,,
\ee

Finally by combining \eqref{eq:statx}, \eqref{eq:aux-x-1} and Lemma \ref{lemma:GLsumma} we get
\be\label{eq:px(px)}
\frac{d}{dp_{x+2}}\frac{p_x(p_{x+2})}{p_{x+2}}<0\,,\quad \frac{d}{dp_{x+2}}\frac{E_{b^{(x+1)},\b^2(p_x(p_{x+2})+p_{x+2})}\left[\tanh^2(\eta)\right]}{p_{x+2}}<0\,. 
\ee

The inequality (\ref{eq:derIIRSx}) is equivalent to $\partial^2_{p_x}\overline\RS^{(x)}\Big|_{p_x=p_x(p_{x+2})}>0$, that is in $p_x(p_{x+2})$ a minimum is attained. So we can set
\be\label{eq:ARSridottak}
\overline\RS^{(x+2)}:=\min_{p_{x}}\overline\RS^{(x)}\,
\ee
and iterate.
\end{proof}

\

To conclude we add the following

\begin{proof}[Proof of Lemma \ref{lemma:repr}]

Let $t\in[0,1]$ and for any $x\in[\nu]$ $\{\eta^x_j\}_{j\in[N_x]}$ be i.i.d. $\mathcal{N}(0,1)$. Recall
\bea
H'_N&:=&-\sum_{x=1}^\nu \sqrt{\sum_{|y-x|=1}\a_y q_y}\sum_{i=1}^{N_x}\eta_i^x\s^{(x)}_i,\\
H_{N,t}&:=&\sqrt tH_N+\sqrt {1-t}H'_N\,.
\eea
Define the interpolating function
\be
\phi_N(t):=\frac1NEE_{0,1}[\log \widehat{E}_{\s^{(1)}\dots \s^{(\nu)}} e^{-\b H_{N,t}}]\,.
\ee
We have
\be
\left\{
\begin{array}{rcl}
\phi_N(0)&=&\sum_{x=1}^\nu\a_xE_{0,1}[\log\cosh\left(b^{(x)}+\b\eta\sqrt{\sum_{y\,:\, |x-y|=1}\a_y q_y}\right)]\\
\phi_N(1)&=&A_N(\b)\,.
\end{array}
\right.
\ee
Recalling the definition (\ref{eq:meanvt}) we compute
$$
\frac{\dd}{\dd t} \phi_N(t)=\frac{\b}{2\sqrt{t}N}\meanv{H_{N}}_t-\frac{\b}{2\sqrt{1-t}N}\meanv{H'_N(t)}_t,
$$
and by Gaussian integration by parts we get
\bea
\meanv{H_{N}}_t&=&\b\frac{\sqrt t}{N}\sum_{|x-y|=1}N_xN_y(1-\meanv{R^{(x)}R^{(y)}}_t),\nn\\
\meanv{H'_N}_t&=&2\b\sqrt{1-t} \sum_{|x-y|=1}N_x\a_yq_y(1-\meanv{R^{(x)}}_t) \nn.
\eea
Therefore
\be
\frac{\dd}{\dd t}\phi_N(t)=\frac{\b^2}{2}\sum_{|x-y|=1}\a_x\a_y(1-q_x)(1-q_y)-\frac{\b^2}{2}\sum_{|x-y|=1}\a_x\a_y\meanv{(q_x-R^{(x)})(q_y- R^{(y)})}_t\,.
\ee
So we can write
\bea
\phi_N(1)-\phi_N(0)&=&\int_0^1\frac{\dd}{\dd t}\phi_N(t)\dd t\nn\\
&=&\frac{\b^2}{2}\sum_{|x-y|=1}\a_x\a_y(1-q_x)(1-q_y)\nn\\
&-&\frac{\b^2}{2}\sum_{|x-y|=1}\a_x\a_y\int_0^1\meanv{(q_x-R^{(x)})(q_y- R^{(y)})}_t\dd t \,.
\eea
Then (\ref{eq:repr}) follows just noting
$$
\RS(q)=\phi_N(0)+\frac{\b^2}{2}\sum_{|x-y|=1}\a_x\a_y(1-q_x)(1-q_y)\,.
$$
\end{proof}

\end{document}